\documentclass[letterpaper, 10 pt, conference]{ieeeconf}  % Comment this line out if you need a4paper

\IEEEoverridecommandlockouts                              % This command is only needed if 
\usepackage{lipsum}
\usepackage{stackrel}
\usepackage{amssymb}
\usepackage{amsmath}
\usepackage[pdftex]{graphicx}
\usepackage{xcolor}
\usepackage{enumerate}
\usepackage{cite}
\usepackage{comment}
\usepackage[capitalize]{cleveref}
\newcommand{\beh}[1]{{\color{red}#1}}

\allowdisplaybreaks[4]

\newtheorem{theorem}{Theorem}

\newtheorem{corollary}{Corollary}
\newtheorem{definition}{Definition}

\newtheorem{lemma}{Lemma}
\newtheorem{assumption}{Assumption}

\newtheorem{remark}{Remark}
\newtheorem{example}{Example}

\DeclareMathOperator{\Equaldef}{\overset{def}{=}}
\DeclareMathOperator{\E}{\EE}
\DeclareMathOperator{\R}{\mathbb{R}}

\def\MA{\mathcal{A}}

\def\bfa{\mathbf{a}}
\def\bfg{\mathbf{g}}
\def\bfq{\mathbf{q}}
\def\bfv{\mathbf{v}}
\def\bfu{\mathbf{u}}
\def\bfY{\mathbf{Y}}
\def\ones{\mathbf{1}}
\def\zeros{\mathbf{0}}
\def\bgamma{\boldsymbol{\gamma}}

\def\EE{\mathbb{E}}
\def\cs{c^{\star}}
\def\as{a^{\star}}
\def\xis{\tau_N^\star}
\title{\LARGE \bf
On the coordination efficiency of strategic multi-agent robotic teams}

\author{Marcos M. Vasconcelos$^{1}$ and Behrouz Touri$^{2}$% <-this % stops a space
\thanks{$^{1}$M. M. Vasconcelos is with the Department of Electrical and Computer Engineering, FAMU-FSU College of Engineering,
        Florida State University, USA. Email:
        {\tt  m.vasconcelos@fsu.edu}}%
\thanks{$^{2}$Behrouz Touri is with the Department of Electrical Engineering, University of California San Diego, USA. Email:
        {\tt btouri@eng.ucsd.edu}}%
}

\begin{document}

\maketitle
\thispagestyle{empty}
\pagestyle{empty}

%%%%%%%%%%%%%%%%%%%%%%%%%%%%%%%%%%%%%%%%%%%%%%%%%%%%%%%%%%%%%%%%%%%%%%%%%%%%%%%%
\begin{abstract}

%This electronic document is a ÒliveÓ template. The various components of your paper [title, text, heads, etc.] are already defined on the style sheet, as illustrated by the portions given in this document.
We study the problem of achieving decentralized coordination by a group of strategic decision makers choosing to engage or not in a task in a stochastic setting. First, we define a class of symmetric utility games that encompass a broad class of coordination games, including the popular framework known as \textit{global games}. With the goal of studying the extent to which agents engaging in a stochastic coordination game indeed coordinate, we propose a new probabilistic measure of coordination efficiency. Then, we provide an universal information theoretic upper bound on the coordination efficiency as a function of the amount of noise in the observation channels. Finally, we revisit a large class of global games, and we illustrate that their Nash equilibrium policies may be less coordination efficient then certainty equivalent policies, despite of them providing better expected utility. This counter-intuitive result, establishes the existence of a nontrivial trade-offs between coordination efficiency and expected utility in coordination games.

%We study the problem of decentralized coordination by a group of strategic decision makers choosing between multiple tasks. We introduce a new global game formulation where the agents decide on two distinct  fundamentals, observed through conditionally  independent channels. We study two cases for the private signals: Gaussian and Poisson. The Gaussian channel is more appropriate to applications in robotics and sensor networks; while the Poisson channel is appropriate for modeling applications in micro-biology and molecular communications. We establish the existence of an intuitive class of coordination policies and measure its effectiveness through a new metric based on Shannon's Entropy.
\end{abstract}

%%%%%%%%%%%%%%%%%%%%%%%%%%%%%%%%%%%%%%%%%%%%%%%%%%%%%%%%%%%%%%%%%%%%%%%%%%%%%%%%
\section{Introduction}

Coordinated behavior is desirable in many distributed autonomous systems such as robotic, social-economic, and biological networks \cite{Arditti:2021,Ramazi:2022,Paarporn:2021,Paarporn:2021b,Soham:2022}. Most of the Engineering literature on coordination assumes that the agents exchange  messages over a communication network to asymptotically agree on a common decision variable, such as in opinion dynamics and distributed optimization. However, in the field of Economics, the topic of coordination has been studied from a different point of view, where the agents do not exchange (explicit) messages but instead act strategically. Such an approach is related to coordination games, in which two or more interacting agents are incentivized to take the same action. Deterministic coordination games are characterized by the existence of multiple equilibria, and often lead to the analysis of social dilemmas. One way to address the multiplicity of equilibria uses a framework known as \textit{global games} \cite{Morris:2003}. 

A global game is a Bayesian coordination game, where each agent plays an action after observing a noisy signal about the state-of-the-world. The state-of-the-world, which we simply refer as \textit{state} captures features such as the strength of the economy in a bank run model, the political regime in a regime change model, or the difficulty of a task in a task-allocation problem. Under certain assumptions on the utility structure, global games admit a unique Bayesian Nash Equilibrium even in the presence of a vanishingly small noise in the observations, resolving the issue of equilibrium selection in games with multiple equilibria \cite{Carlsson:1993}.

\begin{figure}[t!]
    \centering
    \includegraphics[width=0.63\columnwidth]{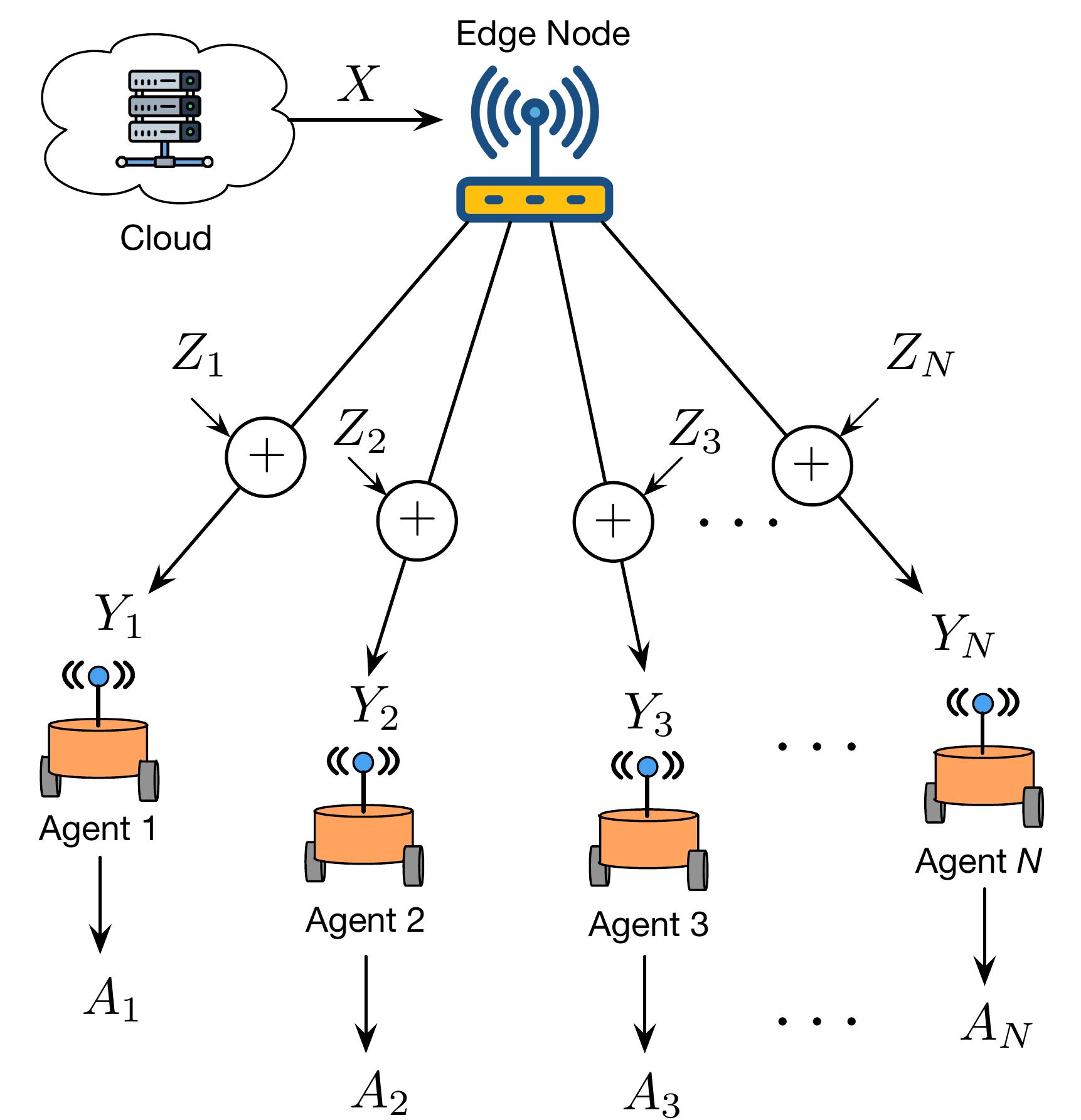}
    \caption{System architecture for the strategic coordination in a robotic team receiving information about a stochastic state variable available at a remote location.}
    \label{fig:system}
\end{figure}
%\vspace{-10pt}
The recent literature in this class of games focuses on aspects related to existence of Nash equilibria in the presence of different information patterns \cite{Dahleh:2016}, or the influence of correlation among the agent's observations \cite{touri2014global,Mahdavifar:2017}, and the impact of different local connectivity patterns in terms of externality in the agents' utility functions \cite{Leister:2022}. Other recent developments look at non-conventional probabilistic models for the noisy signals \cite{Vasconcelos:2022}, and the presence of a multi-dimensional state in a multi-task allocation problem \cite{Wei:2022}.

We consider the federated system architecture outlined in \cref{fig:system}, where the state is available at a remote location (e.g. a cloud server), and broadcast to multiple agents by an edge node or gateway over parallel noisy channels. Upon receiving its noisy signal, an agent makes a binary decision such as to maximize an expected utility function satisfying the strategic complementarity property, leading to coordinated behavior \cite{Hoffmann:2019}. We study how the coordination in a global game degrades with the level of noise in the communication channels. Moreover, we are interested in characterizing the limits of coordination for a given signal to noise ratio used for communication with the robotic agents.
%\subsection{Contributions}

The main contributions of this paper are:

\begin{itemize}
\item We introduce a class of games, namely homogeneous coordination games, that includes a broad class of global games. 

\item We introduce a novel notion of coordination efficiency used to measure the coordination for the homogeneous coordination games.

\item We obtain a fundamental limitation on the coordination efficiency in global games for any policy based on information theoretic tools.
\end{itemize}

%When the utility function satisfies properties known as \textit{strategic complementarity} or \textit{substitutability}, Bayesian Nash equilibria solutions with a threshold structure survive iterated elimination of strictly dominated strategies \cite{Morris:2003}.

%The framework of global games has many potential applications in engineering, such as in task allocation in robotic networks. Importantly and notably, a modeling such engineering applications as global games, lead to a remarkable savings in communication overhead required for consensus/distributed optimization. Moreover, it is a framework well-suited for large-scale systems, which the aforementioned communication overhead turns distributed implementation prohibitive. This lack of communication is compensated by the fact that the agents are strategic and intelligent, meaning that decisions are made taking into account the potential actions of other agents in the system. However, there are many open issues not addressed by the state-of-the-art reported in the economics literature. It is within this context that we situate our proposal.

%\cite{Krishnamurthy:2009}

%\lipsum[1]

%\lipsum[2]

%\lipsum[3]

%\textbf{}\lipsum[4]

%\subsection{Contributions}

%\lipsum[1]

%\subsection{Notation}

%\lipsum[2]

\section{System Model}
In this section, we discuss our model for global games, and in particular, we introduce an important subclass of such games, i.e., \textit{homogeneous coordination} games. 

\subsection{Utility structure}

A global game  is an incomplete information game that is played between $N$ players $[N]\Equaldef\{1,\ldots,N\}$ and nature. Formally, a global game is a tuple $([N],\MA,X,\bfu,\bfY)$, where: 
\begin{enumerate}[(i)]
    \item $\MA=\MA_1\times\MA_2\times \cdots\times \MA_N$ is the joint action set of the $N$ players with $\MA_i$ being the action set for player $i\in [N]$, 
    \item  $X$ is a random variable determining the \textit{type} of nature. We refer to $X$ as the state or the underlying fundamental of the game,  
    \item $\bfu=(u_1,\ldots,u_N):\MA\times \R\to\R^{N}$ is the utility of the $N$ players with $u_i:\MA\times \R\to \R$ being the utility of the $i$-th player that depends on the action of each player and the value of the underlying fundamental $X$, and 
    \item $\bfY=(Y_1,\ldots,Y_N)$ where $Y_i$ is a random variable denoting the player $i\in [N]$ noisy observation (that forms the belief) of the underlying fundamental $X$.
\end{enumerate}
In this work, for any vector 
%(of numbers, functions, etc.) 
$\bfv=(v_1,\ldots,v_N)$ and any $i\in [N]$, we use the notation $\bfv_{-i}=(v_1,\ldots,v_{i-1},v_{i+1},\ldots,v_N)$ and with abuse of notation, we say $\bfv=(v_i,\bfv_{-i})$.  For example, for a joint action $\bfa\in \MA$, we write $\bfa=(a_i,\bfa_{-i})$ for all $i\in [N]$. 

Our work is motivated by the observation that a vast majority of studies in global games and their applications, the underlying games has the following common features:
\begin{enumerate}[(a)]
     \item \label{prop:sym}\textit{Symmetric/Permutation invariant}: In many settings, the utility functions of individual agents are invariant under any permutation of other agents' actions. In other words, $u_i(a_i,\bfa_{-i})=u_i(a_i,\bfa_{-i}P)$ for any $(N-1)\times (N-1)$ permutation matrix\footnote{A matrix is a permutation matrix if all its elements are zero or one and each row and each column has exactly one non-zero element.} $P$. 
     \item \textit{Homogeneous utility functions}: The utility function of the $N$ players are the same in the sense that for any player $i\in [N]$ and any action profile $\bfa\in \MA$, we have $u_i(a_i,\bfa_{-i})=u_1(a_i,\bfa_{-i})$.%\todo{B: I believe this should be $u_1(a_,\bfa_{-1})$} 
      \item \textit{Homogeneous action sets}: In many global games, we are dealing with a large population, and the action set of all players are identical. For example, in the case of political riots, all players decide to take a risky action or safe action in the face of a political regime. In this case, $\MA_i=\{0,1\}$ where $0$ and $1$ correspond to the safe and risky actions, respectively. 
    \item \textit{Coordination promoting}: Again, in most settings of interest, the utility structure of the players is such that it promotes 
    %a form of 
    coordination.
    %between the players. 
    For example, in the case of political uprisings, bank-runs, etc., a well-studied utility function is 
    %the linear function 
    $u_i(\bfa,X)=a_i(\sum_{j=1}^Na_i-X)$. Therefore, in the case, where all players have the perfect information about $X$, i.e., when $Y_i=X$ for all $i\in [N]$, depending on whether $X>\frac{1}{n}$ or $X<\frac{1}{n}$, the only equilibrium of the game is either $\bfa=\ones$ or $\bfa=\zeros$, resulting in coordination among the $N$ players.  
\end{enumerate}

% For a vector of actions $\bfa\in \hMA$ and $i\in [N]$, let us denote the empirical mass distribution of the and , let us define First, note that Property~\ref{prop:sym} implies that the function of 
 Motivated by this, we introduce the notion of \textit{homogeneous coordination} games that formalize a broad class of games satisfying the above properties. For this, let ${\Delta^{k}:=\{\bfq\in \R_+^{k}\mid \sum_{i=1}^{k}q_i=1\}}$ be the probability simplex in $\R^k$. For a finite set $\MA=\{0,\ldots,M-1\}$ and a vector $\bfv\in\MA^d$, where $d\geq 1$, let us define the empirical mass function $G(\bfv)=\big(q_0,\ldots,q_{M-1} \big)\in \Delta^{M}$ by 
\begin{equation}\nonumber
    q_{\ell}= \frac{1}{\gamma}\sum_{j} \mathbf{1}(v_j=\ell), \ \ \ell\in\MA.
\end{equation}
Basically, $q_\ell$ is the proportion of the entries of $\bfv$ that are equal to $\ell\in \MA$. Now, we are ready to formalize the class of homogeneous coordination games.  
%\todo{B: Marcos, I think we don't need to say that all the utilities are the same. It makes the notion slightly more general.}
\vspace{5pt}
% \begin{definition}[Homogeneous Coordination Game] \label{def:coordination} A homogeneous coordination game is a game where all the agents have the same action set $\MA$, the same utility function $u:\MA^N\times\R\to\R $, satisfying the following conditions: 
% \begin{enumerate}[(1)]
%     \item \label{prop:homogsymm} There exists a function $\hat{u}:\MA\times\Delta^{N-1}\times \R\to \R$ where for all $\bfa\in \MA^{N}$, all $i\in [N]$, and all $x\in \R$, we have
%     \begin{align}\label{eqn:symmetry}
%         \hat{u}(a_i,G(\bfa^{-i}),x)=u(a_i,\bfa^{-i},x).
%     \end{align}
%     \item \label{prop:homogcoord }For all $i\in [N]$, all $x \in \mathbb{R}$, and all $\bfg\in\Delta^{N-1}$, there exists an optimal action $a^\star(x)$ such that if the majority of the players are playing  $a^\star(x)$, then player $i$ is better off playing that action. Mathematically, if $g_{a^{\star}(x)}=\max_{\ell\in [M]}g_\ell$, then for all $a_i\in\MA$,
%     \begin{equation}
%        \hat{u}\big(a_i^\star(x),\bfg,x\big) \geq \hat{u}\big(a_i,\bfg,x\big).
%     \end{equation}
%     %\mmv{I believe we need to add this}
% \end{enumerate}
% \end{definition}

\begin{definition}[Homogeneous Coordination Game] \label{def:coordination} {A homogeneous coordination game is a game where all the agents have the same action set $\MA$, the same utility function ${u:\MA^N\times\R\to\R}$, satisfying the following conditions: 
\begin{enumerate}[(1)]
    \item \label{prop:homogsymm} There exists a function $\hat{u}:\MA\times\Delta^{N-1}\times \R\to \R$ where for all $\bfa\in \MA^{N}$, all $i\in [N]$, and all $x\in \R$, we have
    \begin{align}\label{eqn:symmetry}
        \hat{u}(a_i,G(\bfa^{-i}),x)=u(a_i,\bfa^{-i},x).
    \end{align}
    \item \label{prop:homogcoord}For all $i\in [N]$, all $x \in \mathbb{R}$, and all $\bfg\in\Delta^{N-1}$, there exists an optimal action $a^\star(x)$ and majority $c^\star(x)$ of the players, such that if the majority are playing  $a^\star(x)$, then player $i$ is better off playing that action. Mathematically, there exists $c^{\star}(x)\in[0,1]$ such that for any $\bfg\in \Delta^{M}$ with $g_{a^{\star}(x)}=\max_{\ell\in [M]}g_\ell\geq c^{\star}(x)$, we have 
    \begin{equation}\label{eqn:coordination}
       \hat{u}\big(a_i^\star(x),\bfg,x\big) \geq \hat{u}\big(a_i,\bfg,x\big), \qquad \forall a_i\in\MA.
    \end{equation}
    %\mmv{I believe we need to add this}
\end{enumerate}}
\end{definition}

\vspace{5pt}

Note that Property~\eqref{prop:homogsymm} essentially means that the utility function of each player is symmetric/permutation invariant.%\todo{B: Marcos, we can have a little lemma here showing that. Should we?} 
%\mmv{If you have it handy, it would be nice to include it. If not yet, we can save it for a journal paper.}
In other words, for a finite action set $\MA=[M]$, any symmetric/permutation invariant function $f:\MA^{N-1}\to \R$ (as defined in \eqref{prop:sym}), can be written as a function of the empirical mass function of the $M$ actions, i.e., for such utility functions, it does not matter which player is playing what action, but rather how many or what proportion of the players is playing each action. %Therefore, \eqref{eqn:symmetry} essentially requires that we can rewrite the utility function as a function of the empirical mass function. 

For the rest of the paper, with an abuse of notation, instead of $u(a_i,\bfa_{-i},x)$, we may view the utility functions of a homogeneous coordination game to be simply a function of the empirical mass and use the notation $u(a_i,\bfg_{-i},x)$ instead of $\hat{u}(a_i,G(\bfa_{-i}),x)$, where $\bfg_{-i}=G(\bfa_{-i})$. 

%In addition, Property~\ref{prop:homogsymm} means that 

\begin{comment}

\beh{\begin{definition}[Homogeneous Coordination Game] \label{def:coordination} Let all agents observe the state $X=x$ perfectly. A homogeneous coordination game is a game where all the agents have the same utility function, $u$, which satisfies the following: for all $i\in [N]$, $x \in \mathbb{R}$ and $g_{-i}\in\Delta^{N-1}$, there exists an optimal action $a^\star(x)$ such that
\begin{equation}
u\big(a_i^\star,g_{-i},x\big) \geq u\big(a_i,g_{-i},x\big), \ \ a_i\in\mathcal{A},
\end{equation}
and %\mmv{I believe we need to add this}
\begin{equation}
a_i^\star = a^\star(x) = \arg \max_{a\in\mathcal{A}}  \ g_{-i}(a).  
\end{equation}
\end{definition}}

\end{comment}

%\section{System Model}

\subsection{Information structure and policies}
%Consider a system with $N>1$ programmable devices acting as decision-making \textit{agents}. Let $[N]$ denoted the set $\{1,\ldots,N\}$. 
%The main idea is to let autonomous agents decide which task they would like to work on. Each agent which must decide on one out of two tasks to undertake. 
Here, we discuss the assumptions on the fundamental $X$ and individual agents' noisy observation of $Y_i$. Throughout, we assume that $X$ is a zero-mean Gaussian random variable with variance $\sigma_X^2$, i.e., $X\sim\mathcal{N}(0,\sigma_X^2)$. We assume the commonly studied model (cf.~\cite{Morris:2003,touri2014global,Dahleh:2016}) for the $i$-th agent noisy observation $Y_i$ to be $Y_i = X +Z_i$.  We assume that the noise sequence $\{Z_i\}_{i\in [N]}$, is independent and identically distributed across agents and $Z_i \sim \mathcal{N}(0,\sigma^2_{Z})$. Moreover, $\{Z_i\}$ is independent of $X$.

Note that since $X$ and $Y_i$ are jointly Gaussian and the minimum mean squared error estimate of $X$ given $Y_i=y_i$ is linear and is given by
\begin{equation}\label{eq:mmse}
    \hat{x}_{\mathrm{mmse}}(y_i) \Equaldef \E[X\mid Y_i=y_i] = \Big(\frac{\sigma_X^2}{\sigma_X^2+\sigma_Z^2}\Big)y_i .
\end{equation}

In general for games of imperfect information (which includes global games and homogeneous coordination games), the agents take action based on their observation. This leads to the notion of \textit{policy}. For homogeneous coordination games with the action set $\MA=\{1,\ldots, M\}$, a \textit{policy} is a mapping $\gamma_i: \mathbb{R}\rightarrow \mathcal{A}$ that translates agent $i$-s observation to action, i.e., agent $i\in [N]$ takes action $a_i = \gamma_i(Y_i)$.

\subsection{Bayesian Nash Equilibrium}
Let $u_i$ be the utility function of a homogeneous coordination game (\cref{def:coordination}). The agents in this game act in a noncooperative manner, by seeking to maximize their individual \textit{expected} utility with respect to their coordination policies. Let $\boldsymbol{\gamma}\Equaldef (\gamma_1, \ldots, \gamma_N)$ be a \textit{policy profile}%\todo{B: Marcos is there a reason that we use coordination policy instead of policy as defined above?}
, i.e., the collection of policies used by all the agents in the system. %Similarly, let $\boldsymbol{\gamma}_{-i}$ denote the coordination strategy profile used by all the agents except the one from agent $i$.

Given $\boldsymbol{\gamma}_{-i}$, the goal of the $i$-th agent is to solve the stochastic optimization problem
\begin{equation}\nonumber
    \underset{\gamma_i}{\mathrm{maximize}} \ \  \mathcal{J}_i(\gamma_i,\boldsymbol{\gamma}_{-i}) \Equaldef  \EE \Big[ u_i(A_i,G_{-i},X) \Big],
\end{equation}
where the expectation is taken over all the exogenous random variables $X$, and $\{Z_i\}_{i\in [N]}$. This leads to the notion of Bayesian Nash Equilibrium (BNE) strategies.
\begin{definition}[Bayesian Nash Equilibrium]
 A policy profile $\boldsymbol{\gamma}^\star$ is a Bayesian Nash Equilibrium if
\begin{equation}\nonumber
    \mathcal{J}_i(\gamma_i^\star,\gamma_{-i}^\star) \geq 
\mathcal{J}_i(\gamma_i,\boldsymbol{\gamma}_{-i}^\star), \ \mbox{for all } \gamma_i \in \Gamma, \ \ i \in [N],
\end{equation}
where $\Gamma$ is the space of all admissible coordination policies/strategies.
\end{definition}

\subsection{Coordination measure}

Given that the state $X$ is not perfectly observed by the agents, full coordination is often unachievable. In a deterministic setting, defining a precise notion of coordination and agreement is a well-posed problem. However, there are multiple ways of defining a metric of coordination efficiency in a stochastic game setting. One essential feature that such a metric should have is to capture that the extent to which agents coordinate around an optimal action degrades with respect to the amount of noise in the observations. The introduction of framework of homogeneous coordination games, allows us to mathematically define a measure of \textit{coordination efficiency}. %Herein, we define  coordination efficiency as follows.

\begin{definition}[Coordination efficiency]\label{def:coordeffic}
    Let $\bgamma$ be a policy profile of $N$ players in a homogeneous coordination game (as defined in~\cref{def:coordination}). We define the average coordination efficiency $\varrho: \boldsymbol{\gamma} \mapsto [0,1]$ as 
    \begin{equation}\nonumber
        \varrho(\boldsymbol{\gamma}) \Equaldef \frac{1}{N} \sum_{i=1}^N \mathbb{P} \big(\gamma_i(Y_i) = a^\star(X) \big),
    \end{equation}
    where $a^\star(x)$ is the optimal action defined in \cref{def:coordination}.
\end{definition}

\vspace{5pt}

% In the remainder of the paper, we will study a class of homogeneous coordination games and compute the coordination efficiency of a  Nash-Equilibrium solution in the class of threshold policies. 
\begin{comment}
{\color{magenta}
\begin{definition}[System-level coordination efficiency]
Let $\Gamma$ be a coordination policy profile, and $\epsilon$ be a desired level of coordination. Define the system-level coordination efficiency $\delta: \Gamma \mapsto [0,1]$ such that 
\begin{equation}
    \delta(\Gamma;\epsilon) \Equaldef \mathbb{P} \Big( \frac{1}{N} \sum_{i=1}^N \mathbf{1}\big(\gamma_i(Y_i) \neq a^\star(X)\big) \leq \epsilon \Big),
\end{equation}
replace this with
\begin{align}
    C \Equaldef \max\left\{\frac{1}{N} \sum_{i=1}^N \mathbf{1}\big(\gamma_i(Y_i) =\ell\big)\mid {\ell\in \as(X(\omega))}\right\} ,
\end{align}

where $a^\star(x)$ is the optimal action if all the agents could observe the state perfectly.
\end{definition}}
\end{comment}

%\begin{equation}
%h_N(\delta) = \mathbb{P}\Big(\frac{1}{N}\sum_{i=1}^N E_i \leq \delta \Big).
%\end{equation}
%\newpage

\section{Global Games revisited}
An important instance of homogeneous coordination games is a class of binary action global games (i.e., $\MA=\{0,1\}$), where the utility of each agent is given by
%\mmv{Here I am going to typeset a few calculations that I have done a few months ago to gain momentum. Do not pay attention to the notation. I will go back and revise it.}
\begin{equation}\label{eqn:gg}
    u_i(a_i,\bfa_{-i},x) = a_i\cdot\bigg(b\Big(\sum_{j\neq i}a_j\Big) - x\bigg),
\end{equation}
where $b:[0,N-1]\rightarrow \mathbb{R}$ is a continuous and increasing function. The function $b(\cdot)$ is called the \textit{benefit function}. One application for this utility is in distributed task allocation in robotic teams \cite{Kanakia:2016,Berman:2009}, where $x$ represents the difficulty of a task. An agent benefits from engaging in the task if the number of other agents engage in the same action is sufficiently large. However, if the variable $x$ is not perfectly observed by the agents it is not clear whether an agent should engage in the task or not. 

Our next result establishes that in fact global games with utility structure~\eqref{eqn:gg} are homogeneous coordination games. 
\begin{lemma}\label{lemma:benefitiscoordination}
    A global game with $N\geq 2$ players, binary action set $\MA=\{0,1\}$, and utility function~\eqref{eqn:gg} is a homogeneous coordination game for any increasing continuous function $b:[0,N-1]\to\R$. 
\end{lemma}
\begin{proof}
    The homogeneity of the action sets and utility functions follow readily from the definition of such games. To show Property~\eqref{prop:homogsymm}, for $\bfa\in \MA^N$ and $i\in[N]$, let $p=\frac{\sum_{j\neq i}a_j}{N-1}$. Then, $G(\bfa_{-i})=(1-p,p)$. 
    Therefore, letting $\hat{u}(a_i,\bfg,x)\Equaldef a_i(b((N-1)g_1)-x)$ for all $a_i\in\MA$, $\bfg=(g_0,g_1)\in \Delta^2$, and $x\in \R$, we have
    \begin{align}\label{eqn:hgreformulation}
        u_i(a_i,\bfa_{-i},x)&=a_i\Big(b\left((N-1)p\right)-x\Big)=\hat{u}(a_i,G(\bfa_{-i}),x).
    \end{align}
    To show Property~\eqref{prop:homogcoord}, fix $x\in \R$ and let  $b(\cdot)$ be an increasing benefit function. Then, if $x\leq b(0)$, we have 
    \[\hat{u}(1,\bfg,x)= b\Big((N-1)g_1\Big) - x\geq b(0)-x\geq \hat{u}_i(0,\bfg,x)=0.\]
    Therefore, for any $x\leq b(0)$,~\eqref{eqn:coordination} holds with ${\as(x)=1}$ and $\cs(x)=0$. Similarly, it can be shown that for ${x\geq b(N-1)}$, \eqref{eqn:coordination} holds with $\as(x)=0$ and $\cs(x)=0$. For $x\in (b(0),b(N-1))$, we can show that both $\as(x)=0$ and $\as(x)=1$ are possible coordinating actions. To show $\as(x)=1$, let $\cs(x)=\min\{q\in[0,1]\mid b(q(N-1))\geq x\}$ (note that the minimum exists due to the continuity of $b(\cdot)$ and compactness of $[0,1]$). Then for any probability vector $\bfg=(g_0,g_1)\in \Delta^{2}$ with $g_1\geq \cs(x)=q$, we have 
    \begin{align*}
        \hat{u}(1,\bfg,x)&= b\Big((N-1)g_1\Big) - x\cr 
        &\geq b((N-1)q)-x\geq \hat{u}(0,\bfg,x)=0.
    \end{align*}
    Therefore, condition~\eqref{eqn:coordination} holds. Similarly, it can be shown that for $x\in (b(0),b(N-1))$, $\as(x)=0$ is a coordinating action with $\cs(x)=\max\{q\in[0,1]\mid b((1-q)(N-1))\geq x\}$.
\end{proof}

%\todo{B: Should I add a result that the above game is a homogeneous coordination game?}

% Instead, each agent receives an private signal over independent Gaussian channels, which acts as side information about the state $x$. 

%\subsection{Best-response policy}

% In this paper, we will consider a Gaussian prior distribution for the state variable. Let
% \begin{equation}
% X\sim\mathcal{N}(0,\sigma_X^2).
% \end{equation}

% %and assume that 
%\begin{equation}
%Z_i\sim\mathcal{N}(0,\sigma_Z^2), \ \ i\in[N]
%\end{equation}
%where $X \perp \! \! \! \perp \{Z_i\}_{i=1}^N$, and
%\begin{equation}
%Y_i = X + Z_i.
%\end{equation}

%The capacity of each of these channels is
%\begin{equation}
%C=\frac{1}{2}\log_2\Big(1+\frac{\sigma_X^2}{\sigma_Z^2}\Big) \ \text{bits/channel use}.
%\textbf{}\end{equation}

%The non-negative constant $P$ is chosen so that the overall signal-to-noise ratio (SNR) is bounded. Thus,
%\begin{equation}
%\frac{P\sigma_X^2}{\sigma_Z^2} \leq  \textup{SNR}_{\max},
%\end{equation}
%for some $\textup{SNR}_{\max}<\infty$. It is known that 

%Furthermore, notice that when the prior is diffuse, that is, when $\sigma_X^2\rightarrow \infty$, then
%\begin{equation}
%\E[X\mid Y_i=y_i] = y_i.
%\end{equation}

%\mmv{CONTINUE FROM HERE}
In this work, we study the so-called \textit{best-response} policy for the above games. In our case, for a joint policy $\bgamma$, agent $i$-s  best-response policy is 
%\begin{equation}
%    \E \Bigg[ g\Big(\sum_{i\neq j} A_j\Big) - X \ \ \bigg| \ \ Y_i=y_i\Bigg] \stackrel{a_i^ \star = 1}{\geq} 0,
%\end{equation}
%which is equivalent to
\begin{equation}\nonumber
  \mathrm{BR}(y_i,\bgamma_{-i})=
  \begin{cases}1&
    \mbox{if }\begin{aligned}[t]
       &\mathbb{E}[ b\Big( \sum_{i\neq j} \gamma(Y_j) \Big)\mid  Y_i=y_i] \\
       &\qquad\qquad{\geq} \E[X\mid Y_i=y_i]
       \end{aligned}\\
    0&\text{otherwise}
  \end{cases}.
\end{equation}

\begin{comment}

Essentially this utility leads to a best response map that involves the average belief on the action of other agents and the mean square estimate of the fundamental. %I think this is nice because, in the future, we may think of connecting this estimate to the output of a Kalman filter when the fundamental evolves. Future work.
We define the belief as
\begin{equation}
\pi_{ij}(y_i) \Equaldef \mathbb{P}(A_j=1\mid Y_i=y_i).
\end{equation}

\end{comment}

\subsection{Threshold policies and their best-response}

In many games of imperfect information, including global games, we are interested in the class of threshold policies. For global games with binary actions, these are policies where an agent compares its observed signal ${Y_i=y_i}$ to a threshold $\tau_i$ and decides whether to take the risky action ($a_i=1$) or not ($a_i=0$), i.e., 
\begin{equation}\nonumber
    \gamma_i(y_i) = \begin{cases}
    1, \ \ y_i \leq \tau_i\\
    0, \ \ \text{otherwise.}
    \end{cases}
\end{equation}

Using the next result, we will show that the best response to homogeneous threshold policies is a threshold policy.

\begin{lemma}\label{lem:monotonicity}
If the function $b(\cdot)$ is nonnegative and strictly increasing, and all other agents $j\not=i$ utilize a threshold policy $\gamma$ with the same threshold $\tau$, then
\begin{equation}\nonumber
    \EE\Bigg[ b\Big( \sum_{i\neq j}  \gamma(Y_j) \Big) \ \bigg| \ Y_i=y_i\Bigg]
\end{equation}
is a strictly decreasing function of $y_i$.
\end{lemma}

\vspace{5pt}

\begin{proof}
Let $B_{-i}$ denote a function of random variables $\{Y_j\}_{j\not=i}$ given by
$B_{-i} \Equaldef b\Big( \sum_{j\neq i} \gamma(Y_j) \Big)$, with the Cumulative Distribution Function (CDF)
\begin{equation}\nonumber
    F_{B_{-i}\mid Y_i=y_i}(\xi) \Equaldef  \mathbb{P}\Big( B_{-i} \leq \xi \ \Big| \ Y_i = y_i\Big).
\end{equation}

Since  $b(\cdot)\geq 0$, $B_{-i}$ is a nonnegative random variable. Therefore (cf.~\cite[Chapter 1.5, Property E.6]{hajek}),
\begin{equation}\nonumber
    \EE[B_{-i} \mid Y_i=y_i] = \int_{0}^\infty \Big( 1- F_{B_{-i}\mid Y_i=y_i}(\xi)  \Big)d\xi.
\end{equation}

Because the function $b$ is strictly increasing, it admits a unique inverse function $b^{-1}:\mathbb{R}_+\rightarrow \mathbb{R}_+$. Also, since the observations $\{Y_j\}_{j\in [N]}$ are conditionally (on $X$) independent, therefore, 
\begin{multline}\nonumber
    F_{B_{-i}\mid Y_i=y_i}(\xi) = \int_{\mathbb{R}} \mathbb{P} \Big(\sum_{i\neq j}  \gamma(Y_j) \leq b^{-1}(\xi) \ \Big| \ X=x\Big)\\
\times f_{X\mid Y_i=y_i}(x) dx.
\end{multline}

Let $A_j=\gamma(Y_j)$ for $j\not=i$. Conditioned on $X=x$, the collection of Bernoulli random variables $\{A_j \}_{j\neq i}$ is mutually independent with
\begin{equation}\nonumber
\mathbb{P}\big(A_j =1 \mid X=x\big) = \Phi \Big(\frac{\tau_j - x}{\sigma_Z} \Big),
\end{equation}
where $\Phi$ is the CDF of a standard Gaussian random variable\footnote{The CDF of a standard Gaussian random variable is given by
\begin{equation}\nonumber
\Phi(x)\Equaldef\int_{-\infty}^x \frac{1}{\sqrt{2\pi}}\exp\Big(-\frac{\xi^2}{2}\Big)d\xi.
\end{equation}
}. 

Under the assumption of an homogeneous threshold strategy profile where $\tau_j = \tau$ for all $j\neq i$, $\{A_j \}_{j\neq i}$ is identically distributed, which implies that
\begin{equation}\nonumber
\sum_{j\neq i} A_j \mid X=x \sim \mathcal{B} \Bigg(N-1,\Phi \Big(\frac{\tau - x}{\sigma_Z} \Big)\Bigg), 
\end{equation}
where $\mathcal{B}(k,p)$ is a binomial distribution with parameters $(k,p)$. Therefore,%\todo{B: What is $\alpha$ here?}
\begin{multline}\nonumber
    F_{B_{-i}\mid Y_i=y_i}(\xi) = \EE_V \Bigg[ \sum_{\ell = 0}^{\lfloor b^{-1}(\xi)\rfloor} \binom{N-1}{\ell}\Phi \Big(\frac{\tau - V - \alpha y_i}{\sigma_Z} \Big)^{\ell}\\ \times  \Bigg(1-\Phi \Big(\frac{\tau - V - \alpha y_i}{\sigma_Z} \Big)\Bigg)^{N-1-\ell} \Bigg],
\end{multline}
where the expectation is with respect to a random variable $V$ with
\begin{equation}\label{eq:auxiliary}
    V \sim \mathcal{N}\Big(0, \frac{\sigma_X^2\sigma_Z^2}{\sigma_X^2+\sigma_Z^2} \Big) \ \ \text{and} \ \ \alpha\Equaldef \frac{\sigma_X^2}{\sigma_X^2+\sigma_Z^2}.
\end{equation}

%\end{equation}
Let $p(\tau,v,y_i) \Equaldef \Phi \Big(\frac{\tau - v - \alpha y_i}{\sigma_Z} \Big)$. 
Note that $p(\cdot,\cdot,y_i)$ is strictly decreasing in $y_i$, and since the CDF of a binomial random variable computed at a point is a strictly decreasing function %\todo{B: not sure what this means. But generally CDF of discrete R.V.s cannot be strictly increasing} 
in the probability parameter $p$, we have
\begin{equation}\nonumber
\frac{\partial}{\partial y_i}\Big(1-F_{B_{-i}\mid Y_i=y_i}(\xi) \Big) < 0.
\end{equation}

Therefore, 
\begin{equation}\nonumber
\frac{\partial}{\partial y_i} \int_{0}^\infty \Big( 1- F_{B_{-i}\mid Y_i=y_i}(\xi)  \Big)d\xi < 0.
\end{equation}
\end{proof}

\vspace{5pt}

\begin{theorem}\label{thm:threshold}
If the benefit function $b:\R_+\to\R_+$ is nonnegative and strictly increasing, the best-response map to a homogeneous threshold strategy profile is a threshold strategy.
\end{theorem}

\vspace{5pt}

\begin{proof}
Let \[f(y_i)\Equaldef \mathbb{E}[b(\sum_{j\not=i}\gamma(Y_j))\mid Y_i=y_i]-\mathbb{E}[X\mid Y_i=y_i].\] 
\cref{lem:monotonicity} implies that $\mathbb{E}[b(\sum_{j\not=i}\gamma(Y_j))\mid Y_i=y_i]$ is monotonically decreasing in $y_i$ while 
\[\mathbb{E}[X\mid Y_i=y_i]=\Big(\frac{\sigma_X^2}{\sigma_X^2+\sigma_Z^2}\Big)y_i \] is a strictly increasing function of $y_i$. Therefore, $f(y_i)$ is strictly decreasing. Also, since $b$ is an increasing function, $\mathbb{E}[b(\sum_{j\not=i}\gamma(Y_j))\mid Y_i=y_i]\in [b(0),b(N-1)]$ and hence, $\lim_{y_i\to-\infty}f(y_i)=\infty$ and $\lim_{y_i\to \infty}f(y_i)=-\infty$. Therefore, there exists a single crossing point $\bar{\tau}$ such that $f(y_i)>0$ for $y_i<\bar{\tau}$ and $f(y_i)<0$ for $y_i>\bar{\tau}$.
\end{proof}

%Notice that as $\sigma_X^2 \rightarrow \infty$, we have
%\begin{equation}
%X \mid Y_i=y_i \sim \mathcal{N}(y_i,\sigma_Z^2).
% \end{equation}

%The belief depends on the entire conditional probability density of the fundamental $X$ given the observation $Y_i=y_i$, which is:
%\begin{equation}
%X \mid Y_i=y_i \sim \mathcal{N}\Big(\frac{\sigma_X^2}{\sigma_X^2+\sigma_Z^2}y_i, \frac{\sigma_X^2\sigma_Z^2}{\sigma_X^2+\sigma_Z^2} \Big).
%\end{equation}

\subsection{Linear benefit functions}

\Cref{thm:threshold} guarantees that the best response to homogeneous thresholds is a threshold policy for a broad class of Global Games. Once a new threshold is found, other agents imitate by using the same best response threshold. We  recursively use this scheme, which may converge to a BNE policy. 
However, the convergence of such a scheme for an arbitrary benefit function $b(\cdot)$  might not be easy establish, in general. In addition, finding the optimal strategy $\as(x)$ may not be feasible for general benefit functions. However, such characterization is possible for the class of \textit{linear} benefit functions.

Consider the following linear benefit function indexed by $N$, $b_N^{\mathsf{lin}}:\mathbb{R}\rightarrow \mathbb{R}$ such that
\begin{equation}\label{eq:linear}
b_N^{\mathsf{lin}}(\xi) \Equaldef \lambda\cdot\left(\frac{\xi}{N}\right),
\end{equation}
where $\lambda>0$. Define the \textit{belief} function $\pi_{ij}:\mathbb{R}^2 \rightarrow \mathbb{R}$ as
\begin{equation}\nonumber
\pi_{ij}(\tau_j,y_i) \Equaldef \mathbb{P}\big(Y_j \leq \tau \mid Y_i = y_i\big). %\int_{\R} \mathbb{P}(Y_j\leq \tau \mid X=x) f_{X\mid Y_i=y_i}(x)dx.
\end{equation}

%The proof relies on establishing a monotonicity property of the 

% Define $\tilde{\sigma}^2$ as
% \begin{equation}
% \tilde{\sigma}^2 = \frac{\sigma_X^2\sigma_Z^2}{\sigma_X^2+\sigma_Z^2}†
% \end{equation}

After a few algebraic manipulations, we can write 
\begin{equation}\label{eqn:linear}
\pi_{ij}(\tau_j,y_i) = \E \left[\Phi\left(\frac{\tau_j - V -\alpha y_i}{\sigma_Z} \right) \right],
\end{equation}
where $V$ and $\alpha$ are defined in \cref{eq:auxiliary}.

\vspace{5pt}
%\todo{B: Marcos, are the following corollary and remark concerned with linear thresholds?}
\begin{corollary}[Corollary to Theorem 1]\label{lem:BR_linear}
Assuming a linear benefit function (given by~\eqref{eq:linear}), if each agent $j\not=i$ uses a threshold policy with threshold $\tau_j$, then the best-response to any threshold strategy profile is given by the unique solution $\bar{y}_i$ of $ \frac{\lambda}{N}\sum_{i\neq j}\pi_{ij}(\bar{y}_i) = \alpha\bar{y}_i$.
% the following equation:
% \begin{equation}
%  \frac{\lambda}{N}\sum_{i\neq j}\pi_{ij}(\bar{y}_i) = \alpha\bar{y}_i.
% \end{equation}
\end{corollary}

%Recall that our BR map is of the form:
%\begin{equation}
%\frac{\alpha}{N}\sum_{i\neq j}\pi_{ij}(y_i) \stackrel{a_i = 1}{\geq} \Big(\frac{\sigma_X^2}{\sigma_X^2+\sigma_Z^2}\Big)y_i
%\end{equation}

\vspace{5pt}

%\begin{remark}
%The result does not depend on the fact that we are adding beliefs over all $j\neq i$. Therefore, the same result holds if we have a networked Global Game over an extrenality graph.
%\end{remark}

\vspace{5pt}

Based on \cref{lem:BR_linear}, we can define a BR map in the space of threshold policies, which takes a vector of $N$ thresholds and maps into $N$ thresholds. Let $\mathcal{F}:\R^N \rightarrow \R^N$, where
\begin{equation}\nonumber
\mathcal{F}_i(\tau_i,\tau_{-i}) \Equaldef \arg \min_{\xi \in \R} \left( \frac{\lambda}{N}\sum_{i\neq j}\pi_{ij}(\tau_j,\xi) - \alpha\xi\right)^2,
\end{equation}
for all $i\in[N].$

\vspace{5pt}

\begin{remark}
The existence of a Bayesian Nash-equilibrium in threshold policies is easy to show, but whether it is unique depends on establishing a contraction property of $\mathcal{F}$, which is a topic for future work. %, which is not clear to me if it is possible to prove. \#futurework
\end{remark}

\subsection{Homogeneous agents using a threshold $\tau$}

The problem is simpler when we focus only on homogeneous strategy profiles. In that case, the BR to a threshold strategy profile where every agent uses threshold $\tau$ is the unique solution to the following equation
\begin{equation}\label{eq:BR}
    \mathrm{BR}(\tau) = \Big\{\xi^\star : \lambda\frac{(N-1)}{N}\pi(\xi^\star;\tau) = \alpha\xi^\star\Big\},
\end{equation}
here
\begin{equation}\label{eqn:pidef}
    \pi(\xi;\tau) \Equaldef \E \left[\Phi\left(\frac{\tau - \tilde{\sigma}W-\alpha\xi}{\sigma_Z} \right) \right],
\end{equation}
with $W\sim\mathcal{N}(0,1)$ and $\tilde{\sigma}^2=\alpha \sigma_Z^2$.

\vspace{5pt}

\begin{example}
    \Cref{fig:BR} shows the BR function of \cref{eq:BR} and corresponding Nash equilibrium (NE) thresholds for different values of noise variance $\sigma_Z^2$. Two observations from this numerical experiment is that as the noise variance increases, so do the NE thresholds $\tau^\star(\sigma_Z^2)$. Less obvious is the limit of the NE threshold as the variance of the noise goes to zero, that is, with perfect observations. In the noiseless case, \cref{fig:BR} shows that $\tau^\star = 1/2$, which implies that, in this example, $a^ \star(x)=\mathbf{1}(x\leq 0.5).$
\end{example}

\vspace{5pt}

\begin{figure}[t!]
    \centering
\includegraphics[width=0.9\columnwidth]{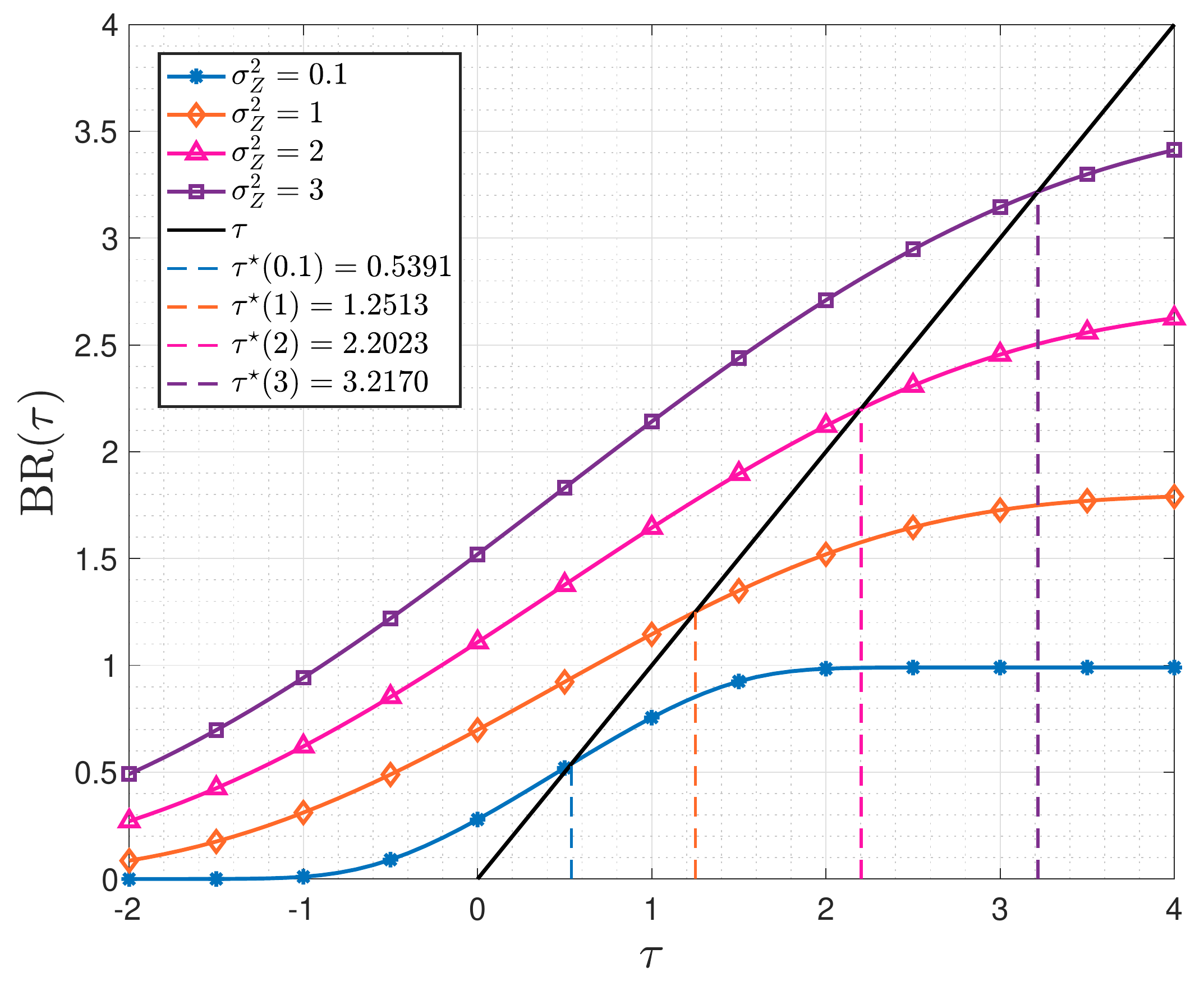}
    \caption{Best response function to a homogeneous strategy profile with threshold $\tau$. Here $N=10$, $\lambda=1$, and $\sigma_X^2=1$.}
    \label{fig:BR}
\end{figure}

First, let us discuss  the following properties of the Gaussian CDF whose proofs are omitted due to space limitations. 

\vspace{5pt}

\begin{lemma}
    Let $V\sim\mathcal{N}(0,1)$. Then, for any $c\in \R$ and any $\epsilon\in \R$, we have 
    \begin{align}\label{eqn:phiprop}
        |\EE_V[\Phi(cV+\epsilon)]-\frac{1}{2}|\leq |\epsilon|.
    \end{align}
    In particular, $\Phi(cV)=\frac{1}{2}$ for all $c\in \R$. 
\end{lemma}

\vspace{5pt}

    Using this result, we can show the following important estimate of the fixed point of~\eqref{eqn:pidef}. 

    \vspace{5pt}
    
\begin{lemma}\label{lemma:xistarestimate}
    For $\sigma_X^2,\sigma_Z^2>0$, let $\xis=\tau_N^\star(\sigma^2_X,\sigma^2_Z)$ be the unique solution to the fixed point equation 
    \begin{align}\label{eqn:xistar}
        \lambda\frac{N-1}{N}\pi(\xis;\xis)=\alpha\xis.
    \end{align} 
    Then, 
    \begin{align}\label{eqn:xistarinequality}
        \frac{1}{2}\leq \left(\frac{\sigma^2_X}{\sigma^2_X+\sigma^2_Z}\right)\frac{\xis}{\lambda (N-1)/N}\leq\left(
        \frac{1}{2}+\frac{\sigma_Z}{\sigma^2_X}\lambda (N-1)/{ N}\right).
    \end{align} 
\end{lemma}
\begin{proof}
    First, note that $0\leq \Phi(\cdot)\leq 1$ and~\eqref{eqn:pidef}, implies that $0\leq \pi(\xis;\xis)\leq 1$. Therefore, the solution to~\eqref{eqn:xistar}, satisfies, 
    \begin{align}\label{eqn:inequalityonxis}
        0\leq \xis\leq \frac{\lambda(N-1)/N}{\alpha}.
    \end{align}
    Using  $\pi(\xis;\xis)=\E \left[\Phi\left(\frac{(1-\alpha)\xis - \tilde{\sigma}W}{\sigma_Z} \right) \right]$, we get 
    \begin{align*}
        \frac{1}{2}&\stackrel{\rm{(a)}}{=}\E \left[\Phi\left(\frac{- \tilde{\sigma}W}{\sigma_Z} \right) \right]\cr 
        &\stackrel{\rm{(b)}}{\leq}
        \E \left[\Phi\left(\frac{(1-\alpha)\xis - \tilde{\sigma}W}{\sigma_Z} \right) \right]\cr 
        &\stackrel{\rm{(c)}}{\leq}
        \E \left[\Phi\left(\frac{(1-\alpha)\frac{\lambda(N-1)/N}{\alpha} - \tilde{\sigma}W}{\sigma_Z} \right) \right]\cr 
        &\stackrel{\rm{(d)}}{\leq}
        \frac{1}{2}+\frac{(1-\alpha)}{\sigma_Z}{\lambda(N-1)/N}.
    \end{align*}
    Here, (a) follows from~\eqref{eqn:phiprop} for $\epsilon=0$, (b) and (c) follow from \cref{eqn:inequalityonxis}, the fact $0<\alpha<1$, the monotonicity of expectation, and the monotonicity of $\Phi(\cdot)$, and (d) follows from~\eqref{eqn:phiprop}. Using the above inequality and the fact that $\xis$ satisfies~\eqref{eqn:xistar}, we have
    \begin{align*}
        \frac{1}{2}\leq \alpha\frac{\xis}{\lambda(N-1)/N}\leq \frac{1}{2}+\frac{(1-\alpha)}{\sigma_Z}{\lambda(N-1)/N}.
    \end{align*}
    Finally, noting $\alpha=\frac{\sigma^2_X}{\sigma^2_X+\sigma^2_Z}$  concludes the proof. 
\end{proof}

\vspace{5pt}
The following result follows immediately from~\cref{lemma:xistarestimate}. 
\begin{theorem}[Diffuse Gaussian priors]\label{thm:diffuse}
    Consider a Global Game with the linear benefit function of \cref{eq:linear}.  Then, for all $\sigma^2_Z>0$, 
    \begin{equation}\nonumber
        \lim_{\sigma^2_X\to\infty}\tau_N^\star = \frac{\lambda}{2}\cdot\Big(1-\frac{1}{N}\Big).
    \end{equation} 
\end{theorem}

\vspace{5pt}

\begin{remark}
In the Economics literature on global games \cite{Carlsson:1993, Morris:2003,Dahleh:2016}, it is customary to assume a diffuse prior distribution on $X$. However, from the Engineering perspective, the assumption of a diffuse Gaussian distribution on $X$ with $\sigma_X^2\rightarrow \infty$ leads to effectively having parallel Gaussian communication channels of  infinite capacity\footnote{The Shannon capacity of a Gaussian channel is given by 
\begin{equation}\nonumber
C=\frac{1}{2}\log_2\Big(1+\frac{\sigma_X^2}{\sigma_Z^2}\Big) \ \text{bits/channel use}.
\end{equation}}. Therefore, from an information theoretic perspective, the effect of the channel noise becomes negligible, leading to perfect estimates of the input $X$ given $Y_i$ in the mean-squared error sense \cite{Guo:2005}.
\end{remark}

% \vspace{5pt}
% \begin{proof}\textit{(Proof of Theorem 2)}
% First we show that when $\sigma_X^2\rightarrow \infty$, then, 
% \begin{equation}\nonumber
% \mathrm{BR}\left(\frac{\lambda}{2}\Big(1-\frac{1}{N}\Big)\right) = \frac{\lambda}{2}\Big(1-\frac{1}{N}\Big).
% \end{equation}
% Taking the limit of the belief, we have 
% \begin{equation}\nonumber
% \lim_{\sigma_X^2\rightarrow \infty} \pi(\xi;\tau) = \E\left[\Phi\left(\frac{\tau-\xi}{\sigma_Z}-W\right) \right].
% \end{equation}
% Then, for $\tau=\lambda(1-1/N)/2$, we have
% \begin{multline}\nonumber
% \lim_{\sigma_X^2\rightarrow \infty} \pi(\xi;\lambda(N-1)/2) \\ =\E\left[\Phi\left(\frac{\lambda(N-1)/2-\xi}{\sigma_Z}-W\right) \right].
% \end{multline}
% If $\xi=\lambda(1-1/N)/2$, we have 
% \begin{equation}
%     \lambda\frac{(N-1)}{N}\E[\Phi(W)] = \lambda\frac{(N-1)}{2N} = \tau^\star_N(\sigma^2_Z).
% \end{equation}

% The second step of the proof consists of showing that for every $\xi<\tau^\star_N(\lambda)$ and $\xi>\tau^\star_N(\sigma^2_Z)$, we have
% \begin{equation}
% \left( \frac{\lambda}{N}\sum_{i\neq j}\pi_{ij}(\tau_j,\xi) - \alpha\xi\right)^2 > 0. 
% \end{equation}
% Due to space constraints, we do not include this step in the present version of this manuscript.
% \end{proof}

% \vspace{5pt}

\begin{remark}
A special case of our result is when $N\rightarrow \infty$  and $\lambda=1$, we get $\tau^\star_{\infty}(1)\to 1/2$, which is established in Morris and Shin \cite{Morris:2003}. 
\end{remark}

\vspace{5pt}

The significance of \cref{thm:diffuse} is that it leads to an Bayesian Nash equilibrium threshold policy corresponding to the case where the signals are observed through perfect channels. Remarkably, for a linear benefit function such policy has a closed form. We refer to this as the \textit{oracle} policy, and is defined as:
\begin{equation}\label{eq:oracle}
\gamma_{\mathrm{oracle}}^\star(x)  \Equaldef \mathbf{1}\big( x \leq \tau^\star_N(\sigma_Z^2)\big).
\end{equation}

%\vspace{5pt}

%\begin{remark}
%Another insight that can be drawn from our analysis is that every result in the global game literature that assumes a degenerate uniform prior distribution over the real line, is effectively assuming channels with infinite capacity. Therefore, the obtained policies are of the ``oracle'' type. 
%\end{remark}

\subsection{Certainty equivalent policies}

The knowledge of the optimal policy for $\sigma_X^2\rightarrow \infty$ motivates the definition of a \textit{certainty equivalent} (CE) policy \cite{Bertsekas:2012} when $\sigma^2_X < \infty$ in which the agents first form a estimate of the fundamental using the MMSE estimator of \cref{eq:mmse}, followed by using the oracle policy in \cref{eq:oracle}: 
\begin{equation}\nonumber
\gamma_{\mathrm{ce}}(y)\Equaldef\mathbf{1}\big( \hat{x}_{\mathrm{mmse}}(y) \leq \tau^\star_N(\sigma_Z^2)\big).
\end{equation}
Thus,
\begin{equation}\nonumber
\gamma_{\mathrm{ce}}(y)=\mathbf{1}\Bigg( y \leq \Big(1+\frac{\sigma^2_Z}{\sigma^2_X}\Big)\cdot\frac{\lambda}{2}\cdot\Big(1-\frac{1}{N}\Big)\Bigg).
\end{equation}

%\newpage 

\section{A fundamental limit on coordination}

\begin{figure}[t!]
    \centering
\includegraphics[width=0.9\columnwidth]{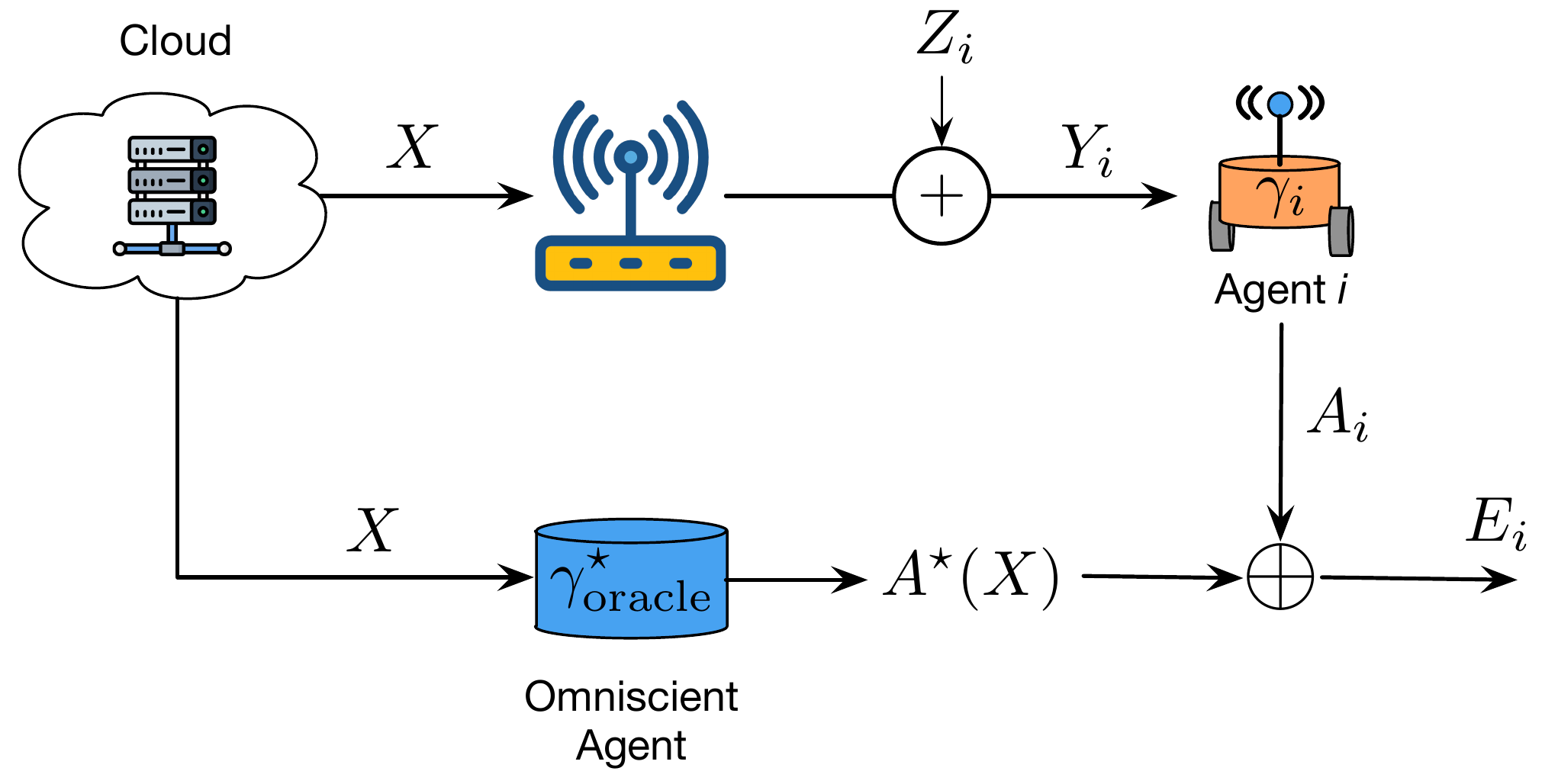}
    \caption{Diagram showing how to compute the coordination error event between a generic agent and an omniscient agent with access to perfect information playing a stochastic coordination game.}
    \label{fig:error}
\end{figure}

We now obtain a universal upper bound on the efficiency of \textit{any} policy, regardless of their structure. Our result is based on \textit{Fano's inequality} \cite{Cover:2006}. Fano's inequality provides a bound on the probability of estimation error of the estimate of a discrete random variable on the basis of side information. 

\vspace{5pt}

\begin{theorem}[Upper bound on coordination efficiency]
For a global game with a linear benefit functions, the coordination efficiency of any homogeneous strategy profile satisfies the following bound
\begin{equation}\nonumber
\varrho \leq 1 - h^{-1}\Big(H\big(A^\star(X) \mid Y_i\big)\Big),
\end{equation}
where $H(\cdot\mid\cdot)$ is the conditional entropy function\footnote{The entropy of a random variable $X\sim f(x)$ is defined as
\begin{equation}\nonumber
H(X) \Equaldef  -\EE\Big[\log_2\big(f(X)\big)\Big].
\end{equation}}, and $h^{-1}(\cdot)$ is the inverse of the binary entropy function over the interval $[0,1/2]$.
\end{theorem}

\vspace{5pt}

\begin{proof}
Let $A^\star(X) \in \{0,1\}$ such that
\begin{equation}\nonumber
A_i^\star(X) = \mathbf{1}\big(X\leq \tau^\star_N(\sigma_Z^2)\big)
\end{equation}
and let $\hat{A}_i(Y)\in \{0,1\}$ denote any estimate of $A^\star(X)$, on the basis of $Y_i$. Then, notice that the following Markov relation is satisfied
\begin{equation}\nonumber
A^\star(X) \leftrightarrow X \leftrightarrow Y \leftrightarrow \hat{A}_i(Y_i).
\end{equation}

Considering the block diagram in \cref{fig:error}, define the error random variable 
\begin{equation}\nonumber
E_i\Equaldef \mathbf{1}\big(\hat{A}_i(Y)\neq A^\star(X)\big),
\end{equation}
and notice that the probability of making an error when  estimating $A^\star(X)$ is at most $1/2$. Fano's inequality \cite{Cover:2006} is a bound on the conditional entropy of the optimal decision computed using the oracle policy given the signal $Y_i$ available to the $i$-th agent
\begin{equation}\nonumber
H\big(A^\star(X) \mid Y_i\big) \leq h(E_i) + \mathbb{P}(E_i=1)\log_2(|\mathcal{A}|-1),
\end{equation}
where $|\mathcal{A}|$ is the cardinality of the decision variable $A_i$. Since out decision variables are binary, we have
\begin{equation}\label{eq:fano}
H\big(A^\star(X) \mid Y_i\big) \leq h(E_i).
\end{equation}

Assuming that we can compute the LHS of \cref{eq:fano}, we obtain a bound on $\mathbb{P}(E_i=1)$, by finding the inverse of the \textit{binary entropy function}\footnote{The binary entropy function is defined as
\begin{equation}\nonumber
h(p) \Equaldef -p\log_2 p - (1-p)\log_2 (1-p).
\end{equation}} within the interval $[0,0.5]$. Finally, notice that for a homogeneous strategy profile the coordination efficiency is
\begin{equation}\nonumber
\varrho = 1 - \mathbb{P}(E_i=1) \leq 1-h^{-1}\Big(H\big(A^\star(X) \mid Y_i\big)\Big).
\end{equation}
\end{proof}

\begin{figure*}[t!]
    \centering
\includegraphics[width=\textwidth]{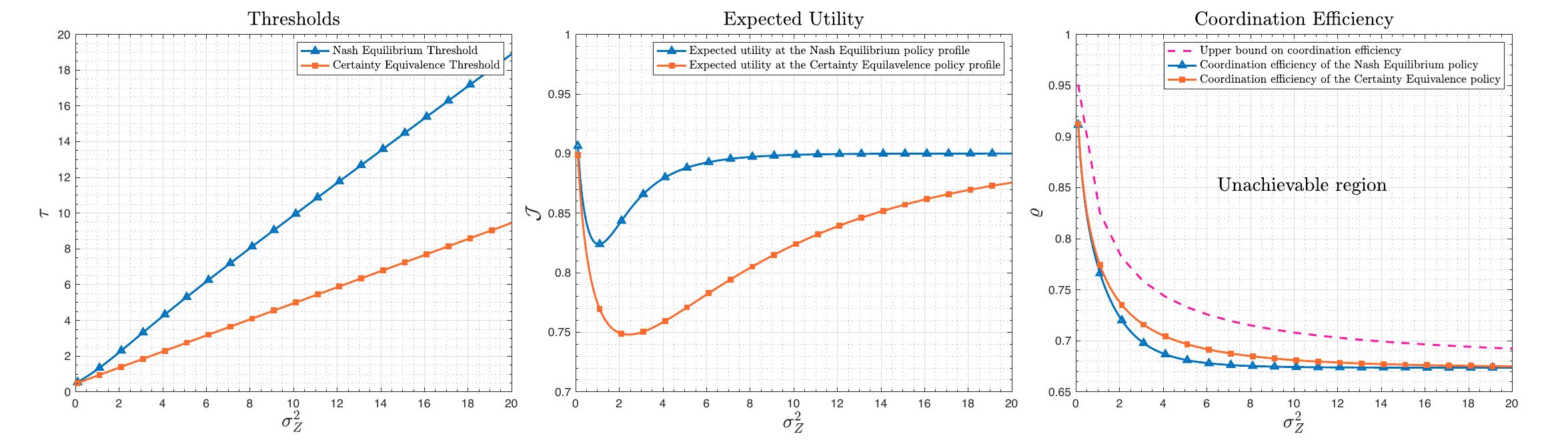}
    \caption{Numerical results for a global game with linear benefit function with $\lambda=1$ for a non-diffuse prior with variance $\sigma_X^2=1$, and a countably infinite number of agents. The NE and CE thresholds as a function of the noise variance (left); The expected utility as a function of the noise variance (center); The coordination efficiency and its universal lower bound (right).}
    \label{fig:efficiency}
\end{figure*}

\subsection{Computing the bound on coordination efficiency}

Using properties of the entropy function, we obtain:
\begin{equation}\nonumber
H\big(A^\star(X) \mid Y_i\big) = H\big(A^\star(X)\big)- H(Y_i) + H\big(Y_i \mid A^\star(X)\big).
\end{equation}
We proceed to compute each of these three terms: the first is the entropy of the optimal decision variable as computed by the oracle:
\begin{equation}\nonumber
h\big(A^\star(X)\big) = h\Bigg(\mathbb{P}\bigg(X\leq \frac{\lambda}{2}\Big(1-\frac{1}{N}\Big) \bigg) \Bigg),
\end{equation}
where $h(\cdot)$ denotes the binary entropy function. 

The second term is the differential entropy of the signal $Y_i$, which is a Gaussian random variable with variance $\sigma^2_X + \sigma^2_Z$. Therefore,
\begin{equation}\nonumber
H(Y_i) = \frac{1}{2}\log_2\big(2\pi e(\sigma^2_X + \sigma^2_Z) \big).
\end{equation}

The third term is more challenging must be computed numerically.
\begin{equation}\nonumber
H(Y_i \mid A^\star(X)=1) =  H \bigg(Y_i \mid X\leq \frac{\lambda}{2}\Big(1-\frac{1}{N}\Big)\bigg). 
\end{equation}
To evaluate this entropy, we must use the conditional probability density function 
\begin{equation}\nonumber
f_{Y_i\mid X\leq \frac{\lambda}{2}(1-\frac{1}{N})}(y_i) = \frac{\int_{-\infty}^{\frac{\lambda}{2}(1-\frac{1}{N})}f_Z(y_i-x)f_X(x)dx}{\int_{-\infty}^{\frac{\lambda}{2}(1-\frac{1}{N})}f_X(x)dx}.
\end{equation}
Similarly,
\begin{equation}\nonumber
H(Y_i \mid A^\star(X)=0) =  H \bigg(Y_i \mid X> \frac{\lambda}{2}\Big(1-\frac{1}{N}\Big)\bigg). 
\end{equation}
To evaluate this entropy, we must use the conditional probability density function 
\begin{equation}\nonumber
f_{Y_i\mid X> \frac{\lambda}{2}(1-\frac{1}{N})}(y_i) = \frac{\int^{\infty}_{\frac{\lambda}{2}(1-\frac{1}{N})}f_Z(y-x)f_X(x)dx}{\int^{\infty}_{\frac{\lambda}{2}(1-\frac{1}{N})}f_X(x)dx}.
\end{equation}
Finally, we can compute:
\begin{multline}\nonumber
H(Y_i \mid A^\star(X)) = H(Y_i \mid A^\star(X)=0)\mathbb{P}\big(A^\star(X)=0\big) \\ + H(Y_i \mid A^\star(X)=1)\mathbb{P}\big(A^\star(X)=1\big),
\end{multline}
where
\begin{multline}\nonumber
H(Y_i \mid A^\star(X)=0) = -\int_{\mathbb{R}} f_{Y_i\mid X> \frac{\lambda}{2}(1-\frac{1}{N})}(y_i) \\ \times \log_2\Big(f_{Y_i\mid X> \frac{\lambda}{2}\big(1-\frac{1}{N})}(y_i)\Big) dy_i
\end{multline}
and
\begin{multline}\nonumber
H(Y_i \mid A^\star(X)=1) = -\int_{\mathbb{R}} f_{Y_i\mid X\leq \frac{\lambda}{2}(1-\frac{1}{N})}(y_i) \\ \times \log_2\Big(f_{Y_i\mid X \leq \frac{\lambda}{2}(1-\frac{1}{N})}(y_i)\Big) dy_i.
\end{multline}

Lastly, the computation of the inverse of the binary entropy function can be efficiently performed numerically.

\section{Numerical results}

The characterization we have provided thus far assumes that a the agents choose their actions according to a policy that ideally tracks the behavior of an omniscient agent that has access to perfect information about the state. Since the agents receive noisy signals about the state, they are not able to perfectly coordinate with the omniscient agent using a threshold policy indexed by $\tau^\star$. 

Assuming that the agents use a generic homogeneous threshold policy indexed by $\tau$, the probability of miscoordination is given $X=x$ is given by:
\begin{multline}
\mathbb{P}(E_i=1\mid X=x) = \bigg(1-\Phi\Big( \frac{\tau-x}{\sigma_Z}\Big)\bigg)\mathbf{1}(x\leq \tau^\star) \\ + \Phi\Big( \frac{\tau-x}{\sigma_Z}\Big)\mathbf{1}(x> \tau^\star).
\end{multline}
Therefore,
\begin{equation}
\varrho(\boldsymbol{\gamma}) = 1 -\int_{\mathbb{R}}\mathbb{P}(E_i=1\mid X=x)f(x)dx.
\end{equation}

Assume a global game with linear benefit function, and a number of agents $N\rightarrow \infty$. The optimal threshold used by the omniscient agent is $\tau^\star = \lambda/2$. Form the agent standpoint, we consider two strategies: 1. computing the NE threshold for the global game, using the prior information $\sigma_X^2$ and $\sigma_Z^2$, and the parameter $\lambda$; 2. estimate the state variable $X$ using a MMSE estimator and using the certainty equivalent policy. 

\Cref{fig:efficiency} (left) shows the thresholds corresponding to these two types of coordination policies for a system $\sigma_X^2=1$ and $\lambda=1$ as a function of the noise variance $\sigma_Z^2$. We can clearly see how different these two policies are. Moreover, there is also a larger computational cost of solving for the NE in the first strategy, whereas the CE strategy can be obtained in closed form in this case. \Cref{fig:efficiency} (center) shows the expected utility of these two strategies. There is a substantial gap between the utilities of an agent using the NE versus CE. This is also clear, because CE in stochastic control and optimization is a suboptimal strategy, in general. More surprisingly is the fact that CE yields a better coordination efficiency, as shown in \cref{fig:efficiency} (right). %Although the gain margin is not large at an individual agent level, it may lead to a large system level difference when other Bayesian metrics of coordination are considered. 

Finally, \cref{fig:efficiency} (right) also shows the information theoretic upper bound on coordination efficiency for any homogeneous policy profile (not just threshold policies). The significance of this figure is that it establishes that certain coordination efficiencies cannot be achieved by any policy for a given level of noise in the communication channel between the gateway and the robotic agents, in a practical application. Therefore, when properly planning for a distributed implementation of a collective task performed by strategic self-interested agents, the system designer needs to communicate at a certain signal to noise ratio, which is not determined by the bit error rate at the receiver, but instead by the level of collective coordination it is interested in achieving.

%Moreover, at a system level, when there is partial information, the ability of the system to achieve a coordinated behavior should degrade when the number of agents increase. In this section we provide numerical evidence of this degradation.
\begin{comment}
Let $h_N(\delta)$ be a function defined as:
\begin{equation}
h_N(\delta) = \mathbb{P}\Big(\frac{1}{N}\sum_{i=1}^N E_i \leq \delta \Big).
\end{equation}

\begin{equation}
h_N(\delta) = \int_\mathbb{R}  \mathbb{P}\Big(\frac{1}{N}\sum_{i=1}^N E_i \leq \delta \mid X=x\Big)f_X(x)dx
\end{equation}

\begin{multline}
 \mathbb{P}\Big(\frac{1}{N}\sum_{i=1}^N E_i \leq \delta \mid X=x\Big) = \sum_{k=0}^{\lfloor N\delta \rfloor} \binom{N}{k} \big(g_{\tau}(x)\big)^k\\ \times \big(1-g_{\tau}(x)\big)^{N-k},
\end{multline}
where $g_{\tau}(x)$ is the probability of coordination error at the agent level, when an agent uses a threshold policy parameterized by $\tau$, conditioned on $X=x$, i.e.,
\begin{multline}
g_{\tau}(x) \Equaldef \bigg(1-\Phi\Big( \frac{\tau-x}{\sigma_Z}\Big)\bigg)\mathbf{1}(x\leq \tau^\star) \\ + \Phi\Big( \frac{\tau-x}{\sigma_Z}\Big)\mathbf{1}(x> \tau^\star),
\end{multline}
where $\tau^\star$ is the threshold of the omniscient agent.
\end{comment}

\ \

\section{Conclusions and Future work}

%Strategic coordination is a topic with applications in many fields, ranging from Economics to Engineering. However, in a stochastic setting, there is a gap in the existing literature on measuring how much coordination a system of self-interested agents operating under uncertainty can achieve. Herein, we considered a simple federated setup where a random state is broadcast to a set of agents, which engage in a coordination game. 

We defined the class of homogeneous coordination games which encompass the popular class of global games. Then, we proposed a Bayesian metric of coordination based on the probabilities that the agents will take the ``right'' action by aligning their decisions with the ones from an omniscient agent with access to the perfect state of the system. We show that this metric of coordination efficiency can be bounded using information theoretic inequalities, establishing regimes in which certain levels of coordination are impossible to achieve. To the best of our knowledge, this is the first time  such methods are used in the context of global games.

Future work on this topic will include design of new learning algorithms (for threshold policies) in the presence of local data at the agents, the presence of partially connected influence graphs on the agent's benefit functions, and the characterization of better upper bounds on coordination efficiency that would take into account the structure of the policy (e.g. threshold).

%\cite{Carlsson:1993}

\bibliographystyle{IEEEtran}

\bibliography{IEEEabrv,MTGG}

\end{document}